\documentclass[11pt]{article}
\newcommand{\taille}{}
\usepackage{latexsym} 
\usepackage{theorem}
\usepackage{graphics}
\usepackage{graphicx}
\usepackage{amsmath}
\usepackage{amssymb}

\usepackage{placeins} 

\newtheorem{defeng}{Definition}[section]
\newtheorem{theorem}[defeng]{Theorem}

\newtheorem{lemma}[defeng]{Lemma}

{\theorembodyfont{\rmfamily} }
{\theorembodyfont{\rmfamily} }
{\theorembodyfont{\rmfamily} }
{\theoremstyle{break}\theorembodyfont{\rmfamily} }
{\theoremstyle{break}\theorembodyfont{\rmfamily} }

\newcommand{\tp}{\!-\!}

\newcounter{claim}

\newenvironment{proof}[1][]%
 {\noindent {\setcounter{claim}{0}\sc proof ---
   }{#1}{}}{\hfill$\Box$\vspace{2ex}} 

\newenvironment{claim}[1][]%
{\refstepcounter{claim}\vspace{1ex}\noindent{(\it\arabic{claim}){#1}{}}\it}{\vspace{1ex}}

\newenvironment{proofclaim}[1][]%
	{\noindent {}{#1}{}}{ This proves~(\arabic{claim}).\vspace{1ex}}

\title{Detecting induced subgraphs}
\date{October 9, 2007\\Revised June 2, 2008}

\author{
  Benjamin L\'ev\^eque\thanks{CNRS, Laboratoire G-SCOP, 46 Avenue F\'{e}lix
    Viallet, 38031~Grenoble~Cedex,
    France. (benjamin.leveque@g-scop.inpg.fr,
    Frederic.Maffray@g-scop.inpg.fr)} , 
  David Y. Lin\thanks{Princeton
    University, Princeton, NJ, 08544. dylin@princeton.edu} ,
  Fr\'ed\'eric Maffray$^{*}$, 
  Nicolas Trotignon\thanks{CNRS, Universit\'e Paris 7, Paris Diderot, LIAFA, Case
    7014, 75205 Paris Cedex 13,   France.   E-mail:
    nicolas.trotignon@liafa.jussieu.fr.\newline This work has
    been partially supported by ADONET network, a Marie Curie training
    network of the European Community.}}
    
\begin{document}

\maketitle

\abstract{An \emph{s-graph} is a graph with two kinds of edges:
  \emph{subdivisible} edges and \emph{real} edges.  A
  \emph{realisation} of an s-graph $B$ is any graph obtained by
  subdividing subdivisible edges of $B$ into paths of arbitrary length
  (at least one).  Given an s-graph $B$, we study the decision problem
  $\Pi_B$ whose instance is a graph $G$ and question is ``Does $G$
  contain a realisation of $B$ as an induced subgraph?''.  For
  several $B$'s, the complexity of $\Pi_B$ is known and here we give
  the complexity for several more.

Our NP-completeness proofs for $\Pi_B$'s rely on the NP-completeness
proof of the following problem.  Let $\cal S$ be a set of graphs and
$d$ be an integer.  Let $\Gamma_{\cal S}^d$ be the problem whose
instance is $(G, x, y)$ where $G$ is a graph whose maximum degree is
at most $d$, with no induced subgraph in $\cal S$ and $x, y \in V(G)$
are two non-adjacent vertices of degree~$2$.  The question is ``Does
$G$ contain an induced cycle passing through $x, y$?''.  Among several
results, we prove that $\Gamma^3_{\emptyset}$ is NP-complete.  We give
a simple criterion on a connected graph $H$ to decide whether
$\Gamma^{+\infty}_{\{H\}}$ is polynomial or NP-complete.  The
polynomial cases rely on the algorithm three-in-a-tree, due to
Chudnovsky and Seymour.}

\noindent{\bf AMS Mathematics Subject Classification:} 05C85, 68R10,
68W05, 90C35

\noindent {\bf Key words}: detecting, induced, subgraphs

\taille

\section{Introduction}

In this paper graphs are simple and finite.  A \emph{subdivisible
  graph} (\emph{s-graph} for short) is a triple $B = (V, D, F)$ such
that $(V, D \cup F)$ is a graph and $D \cap F = \emptyset$.  The edges
in $D$ are said to be \emph{real edges of $B$} while the edges in $F$
are said to be \emph{subdivisible edges of $B$}.  A \emph{realisation}
of $B$ is a graph obtained from $B$ by subdividing edges of $F$ into
paths of arbitrary length (at least one).  The problem $\Pi_B$ is the
decision problem whose input is a graph $G$ and whose question is
"Does $G$ contain a realisation of $B$ as an induced subgraph?''.  On
figures, we depict real edges of an s-graph with straight lines, and
subdivisible edges with dashed lines.

\begin{figure}
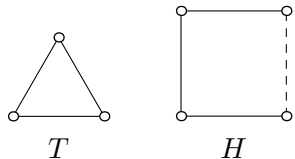

\center
\begin{tabular}{ccc}
\includegraphics{bigraphs.19}&
\hspace{10em}&
\includegraphics{bigraphs.20}\\
$T$&&
$H$
\end{tabular}
\caption{s-graphs yielding trivially polynomial
  problems\label{fig:trivial}}
\end{figure}

Several interesting instance of $\Pi_B$ are studied in the
literature.  For some of them, the existence of a polynomial time
algorithm is trivial, but efforts are devoted toward optimized
algorithms.  For example, Alon, Yuster and Zwick \cite{alon:triangle}
solve $\Pi_T$ in time $O(m^{1.41})$ (instead of the obvious $O(n^3)$
algorithm), where $T$ is the s-graph depicted on
Figure~\ref{fig:trivial}.  This problem is known as \emph{triangle
detection}.  Rose, Tarjan and Lueker~\cite{RTL76} solve
$\Pi_{H}$ in time $O(n+m)$ where $H$ is the s-graph depicted on
Figure~\ref{fig:trivial}.


\begin{figure}
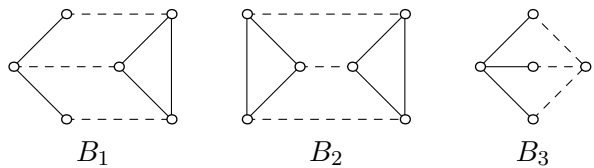

\center
\begin{tabular}{ccccc}
\includegraphics{bigraphs.17}&
\hspace{10em}&
\includegraphics{bigraphs.16}&
\hspace{10em}&
\includegraphics{bigraphs.18}\\
$B_1$&&
$B_2$&&
$B_3$
\end{tabular}
\caption{Pyramids, prisms and thetas\label{fig:ppt}}
\end{figure}

For some $\Pi_B$'s, the existence of a polynomial time algorithm is
non-trivial.  A \emph{pyramid} (resp.  \emph{prism}, \emph{theta}) is
any realisation of the s-graph $B_1$ (resp.  $B_2$, $B_3$) depicted on
Figure~\ref{fig:ppt}.  Chudnovsky and
Seymour~\cite{chudnovsky.c.l.s.v:reco} gave an $O(n^9)$-time algorithm
for $\Pi_{B_1}$ (or equivalently, for detecting a pyramid).  As far as
we know, that is the first example of a solution to a $\Pi_B$ whose
complexity is non-trivial to settle.  In contrast, Maffray and
Trotignon~\cite{maffray.t:reco} proved that $\Pi_{B_2}$ (or detecting
a prism) is NP-complete.  Chudnovsky and
Seymour~\cite{chudnovsky.seymour:theta} gave an $O(n^{11})$-time
algorithm for $P_{B_3}$ (or detecting a theta).  Their algorithm
relies on the solution of a problem called ``three-in-a-tree'', that
we will define precisely and use in Section~\ref{sec:bienstock}.  The
three-in-tree algorithm is quite general since it can be used to solve
a lot of $\Pi_B$ problems, including the detection of pyramids.

These facts are a motivation for a systematic study of $\Pi_B$.  A
further motivation is that very similar s-graphs can lead to a
drastically different complexity.  The following example may be more
striking than pyramid/prism/theta~: $\Pi_{B_4}, \Pi_{B_6}$ are
polynomial and $\Pi_{B_5}, \Pi_{B_7}$ are NP-complete, where $B_4,
\dots, B_7$ are the s-graphs depicted on Figure~\ref{fig:antenna}.
This will be proved in section~\ref{sec:antenna}.

\begin{figure}
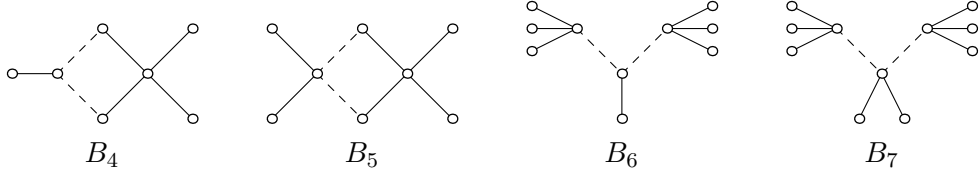

\center
\begin{tabular}{ccccccc}
\includegraphics{bigraphs.15}&
\hspace{10em}&
\includegraphics{bigraphs.14}&
\hspace{10em}&
\includegraphics{bigraphs.22}&
\hspace{10em}&
\includegraphics{bigraphs.21}\\
$B_4$&&
$B_5$&&
$B_6$&&
$B_7$
\end{tabular}
\caption{Some s-graphs with pending edges\label{fig:antenna}}
\end{figure}

\subsection*{Notation and remarks}

By $C_k$ ($k\geq 3$) we denote the cycle on $k$ vertices, by $K_l$
($l\geq 1$) the clique on $l$ vertices.  A \emph{hole} in a graph is
an induced cycle on at least four vertices.  We denote by $I_l$
($l\geq 1$) the tree on $l+5$ vertices obtained by taking a path of
length $l$ with ends $a, b$, and adding four vertices, two of them
adjacent to $a$, the other two to $b$; see Figure~\ref{fig:I1}.  When
a graph $G$ contains a graph isomorphic to $H$ as an induced subgraph,
we will often say ``$G$ contains an $H$''.

\begin{figure}
  \center
  \includegraphics{smallgraphs.7}

  \caption{$I_1$\label{fig:I1}}  
\end{figure}

Let $(V, D, F)$ be an s-graph.  Suppose that $(V, D \cup F)$ has a
vertex of degree one incident to an edge $e$.  Then $\Pi_{(V, D
  \cup\{e\}, F \setminus\{e\})}$ and $\Pi_{(V,D\setminus\{e\},
  F\cup\{e\})}$ have the same complexity, because a graph $G$ contains
a realisation of $(V, D \cup\{e\}, F\setminus\{e\})$ if and only if it
contains a realisation of $(V, D\setminus\{e\}, F\cup\{e\})$.  For the
same reason, if $(V, D \cup F)$ has a vertex of degree two incident to
the edges $e\neq f$ then $\Pi_{(V, D \setminus \{e\} \cup\{f\}, F
  \setminus\{f\} \cup \{e\})}$, $\Pi_{(V, D \setminus \{f\} \cup\{e\},
  F \setminus\{e\} \cup \{f\})}$ and $\Pi_{(V, D \setminus \{e, f\}, F
  \cup \{e, f\})}$ have the same complexity.  If $|F| \leq 1$ then
$\Pi_{(V, D, F)}$ is clearly polynomial.  Thus in the rest of the
paper, we will consider only s-graphs $(V, D, F)$ such that:

\begin{itemize}
\item 
  $|F| \geq 2$;
\item 
  no vertex of degree one is incident to an edge of $F$;
\item
  every induced path of $(V, D\cup F)$ with all interior vertices of
  degree~2 and whose ends have degree $\neq 2$ has at most one edge in
  $F$. Moreover, this edge is incident to an end of the path;
\item
  every induced cycle with at most one vertex $v$ of degree at least
  $3$ in $(V, D\cup F)$ has at most one edge in $F$ and this edge is
  incident to $v$ if $v$ exists (if it does not then the cycle is a
  component of $(V, D\cup F)$).
\end{itemize}

\section{Detection of holes with prescribed vertices}
\label{sec:bienstock}

Let $\Delta(G)$ be the maximum degree of $G$.  Let $\cal S$ be a set
of graphs and $d$ be an integer.  Let $\Gamma_{\cal S}^d$ be the
problem whose instance is $(G, x, y)$ where $G$ is a graph such that
$\Delta(G) \leq d$, with no induced subgraph in $\cal S$ and $x, y \in
V(G)$ are two non-adjacent vertices of degree~$2$.  The question is
``Does $G$ contain a hole passing through $x, y$?''.  For simplicity,
we write $\Gamma_{\cal S}$ instead of $\Gamma_{\cal S}^{+\infty}$ (so,
the graph in the instance of $\Gamma_{\cal S}$ has unbounded degree).
Also we write $\Gamma^d$ instead of $\Gamma^d_{\emptyset}$ (so the
graph in the instance of $\Gamma^d$ has no restriction on its induced
subgraphs).  Bienstock~\cite{bienstock:evenpair} proved that $\Gamma =
\Gamma_{\emptyset}$ is NP-complete.  For ${\cal S} = \{K_3\}$ and
${\cal S} = \{K_{1, 4}\}$, $\Gamma_{\cal S}$ can be shown to be
NP-complete, and a consequence is the NP-completeness of several
problems of interest: see~\cite{maffray.t:reco}
and~\cite{maffray.t.v:3pcsquare}.

In this section, we try to settle $\Gamma^d_{\cal S}$ for as many
$\cal S$'s and $d$'s as we can.  In particular, we give the complexity
of $\Gamma_{\cal S}$ when $\cal S$ contains only one connected graph
and of $\Gamma^d$ for all $d$.  We also settle $\Gamma^d_{\cal S}$ for
some cases when $\cal S$ is a set of cycles. The polynomial cases are
either trivial, or are a direct consequence of an algorithm of
Chudnovsky and Seymour. The NP-complete cases follow from several
extensions of Bienstock's construction.

\subsection{Polynomial cases}

Chudnovsky and Seymour~\cite{chudnovsky.seymour:theta} proved that the
problem whose instance is a graph $G$ together with three vertices $a,
b, c$ and whose question is "Does $G$ contain a tree passing through
$a, b, c$ as an induced subgraph?" can be solved in time $O(n^4)$. We
call this algorithm ``three-in-a-tree''. Three-in-a-tree can be used
directly to solve $\Gamma_{\cal S}$ for several $\cal S$'s.  Let us
call \emph{subdivided claw} any tree with one vertex $u$ of degree~3,
three vertices $v_1, v_2, v_3$ of degree~1 and all the other vertices
of degree~2.

\begin{theorem}
  \label{th:clawpoly}
Let $H$ be a graph on $k$ vertices that is either a path or a
subdivided claw.  There is an $O(n^{k})$-time algorithm for
$\Gamma_{\{H\}}$.
\end{theorem}

\begin{proof}
Here is an algorithm for $\Gamma_{\{H\}}$.  Let $(G, x, y)$ be an
instance of $\Gamma_H$.  If $H$ is a path on $k$ vertices then every
hole in $G$ is on at most $k$ vertices.  Hence, by a brute-force
search on every $k$-tuple, we will find a hole through $x, y$ if there
is any.  Now we suppose that $H$ is a subdivided claw.  So $k\ge 4$.
For convenience, we put $x_1=x$, $y_1=y$.  Let $x_0, x_2$ (resp.
$y_0, y_2$) be the two neighbors of $x_1$ (resp.  $y_1$).  

First check whether there is in $G$ a hole $C$ through $x_1, y_1$ such
that the distance between $x_1$ and $y_1$ in $C$ is at most $k-2$.  If
$k=4$ or $k=5$ then $\{x_0, x_1, x_2, y_0, y_1, y_2\}$ either induces
a hole (that we output) or a path $P$ that is contained in every hole
through $x,y$.  In this last case, the existence of a hole through $x,
y$ can be decided in linear time by deleting the interior of $P$, deleting the
neighbors in $G\setminus P$ of the interior vertices of $P$ and by checking
the connectivity of the resulting graph.  Now suppose $k\ge 6$.  For
every $l$-tuple $(x_3, \dots, x_{l+2})$ of vertices of $G$, with
$l\leq k-5$, test whether $P = x_0 \tp x_1 \tp \cdots \tp x_{l+2} \tp
y_2 \tp y_1 \tp y_0$ is an induced path, and if so delete the interior
vertices of $P$ and their neighbors except $x_0, y_0$, and look for a
shortest path from $x_0$ to $y_0$.  This will find the desired hole if
there is one, after possibly swapping $x_0, x_2$ and doing the work
again.  This takes time $O(n^{k-3})$.

Now we may assume that in every hole through $x_1,y_1$, the distance
between $x_1, y_1$ is at least $k-1$.

Let $k_i$ be the length of the unique path of $H$ from $u$ to $v_i$,
$i=1, 2, 3$.  Note that $k = k_1 + k_2 + k_3 + 1$.  Let us check every
$(k-4)$-tuple $z = (x_3, \dots, x_{k_1+1}, y_3, \dots, y_{k_2 + k_3})$
of vertices of $G$.  For such a $(k-4)$-tuple, test whether $x_0 \tp
x_1 \tp \cdots \tp x_{k_1+1}$ and $P = y_0 \tp y_1 \tp \cdots \tp
y_{k_2 + k_3}$ are induced paths of $G$ with no edge between them
except possibly $x_{k_1 + 1}y_{k_2 + k_3}$.  If not, go to the next
$(k-4)$-tuple, but if yes, delete the interior vertices of $P$ and
their neighbors except $y_0, y_{k_2+k_3}$.  Also delete the neighbors
of $x_2, \dots, x_{k_1}$, except $x_1, x_2, \dots, x_{k_1},
x_{k_1+1}$.  Call $G_z$ the resulting graph and run three-in-a-tree in
$G_z$ for the vertices $x_1, y_{k_2 + k_3}, y_0$.  We claim that the
answer to three-in-a-tree is YES for some $(k-4)$-tuple if and only if
$G$ contains a hole through $x_1, y_1$ (after possibly swapping $x_0,
x_2$ and doing the work again).

To prove this, first assume that $G$ contains a hole $C$ through $x_1,
y_1$ then up to a symmetry this hole visits $x_0, x_1, x_2, y_2, y_1,
y_0$ in this order.  Let us name $x_3, \dots, x_{k_1+1}$ the vertices
of $C$ that follow after $x_1, x_2$ (in this order), and let us name
$y_3, \dots, y_{k_2+k_3}$ those that follow after $y_1, y_2$ (in
reverse order).  Note that all these vertices exist and are pairwise
distinct since in every hole through $x_1, y_1$ the distance between
$x_1, y_1$ is at least $k-1$.  So the path from $y_0$ to $y_{k_2+k_3}$
in $C\setminus y_1$ is a tree of $G_z$ passing through $x_1, y_{k_2 +
  k_3}, y_0$, where $z$ is the $(k-4)$-tuple $(x_3, \dots, x_{k_1+1},
y_3, \dots, y_{k_2 + k_3})$.

Conversely, suppose that $G_z$ contains a tree $T$ passing through
$x_1, y_{k_2 + k_3}, y_0$, for some $(k-4)$-tuple $z$.  We suppose
that $T$ is vertex-inclusion-wise minimal.  If $T$ is a path visiting
$y_0, x_1, y_{k_2 + k_3}$ in this order, then we obtain the desired
hole of $G$ by adding $y_1, y_2, \dots, y_{k_2+k_3-1}$ to $T$.  If $T$
is a path visiting $x_1, y_0, y_{k_2 + k_3}$ in this order, then we
denote by $y_{k_2+k_3+1}$ the neighbor of $y_{k_2 + k_3}$ along $T$.
Note that $T$ contains either $x_0$ or $x_2$.  If $T$ contains $x_0$,
then there are three paths in $G$: $y_0 \tp T \tp x_0 \tp x_1 \tp
\cdots \tp x_{k_1}$, $y_0 \tp T \tp y_{k_2+k_3+1} \tp \cdots \tp
y_{k_3+2}$ and $y_0 \tp y_1 \tp \cdots \tp y_{k_3}$.  These three
paths form a subdivided claw centered at $y_0$ that is long enough to
contain an induced subgraph isomorphic to $H$, a contradiction.  If
$T$ contains $x_{2}$ then the proof works similarly with $y_0 \tp T
\tp x_{k_1+1} \tp x_{k_1} \tp \cdots \tp x_{1}$ instead of $y_0 \tp T
\tp x_0 \tp x_1 \tp \cdots \tp x_{k_1}$.  If $T$ is a path visiting
$x_1, y_{k_2 + k_3}, y_0$ in this order, the proof is similar, except
that we find a subdivided claw centered at $y_{k_2+k_3}$.  If $T$ is
not a path, then it is a subdivided claw centered at a vertex $u$ of
$G$.  We obtain again an induced subgraph of $G$ isomorphic to $H$ by
adding to $T$ sufficiently many vertices of $\{x_0, \dots x_{k_1+1},
y_0, \dots, y_{k_2+k_3}\}$.
\end{proof}

\subsection{NP-complete cases (unbounded degree)}

Many NP-completeness results can be proved by adapting Bienstock's
construction.  We give here several polynomial reductions from the
problem {\sc $3$-Satisfiability} of Boolean functions.  These results
are given in a framework that involves a few parameters, so that our
result can possibly be used for different problems of the same type.
Recall that a Boolean function with $n$ variables is a mapping $f$
from $\{0, 1\}^n$ to $\{0, 1\}$.  A Boolean vector $\xi\in\{0, 1\}^n$
is a \emph{truth assignment satisfying $f$} if $f(\xi)=1$.  For any
Boolean variable $z$ on $\{0, 1\}$, we write $\overline{z}:=1-z$, and
each of $z, \overline{z}$ is called a \emph{literal}.  An instance of
{\sc $3$-Satisfiability} is a Boolean function $f$ given as a product
of clauses, each clause being the Boolean sum $\vee$ of three
literals; the question is whether $f$ is satisfied by a truth
assignment.  The NP-completeness of {\sc $3$-Satisfiability} is a
fundamental result in complexity theory, see \cite{garey.johnson:np}.

Let $f$ be an instance of {\sc $3$-Satisfiability}, consisting of $m$
clauses $C_1, \ldots, C_m$ on $n$ variables $z_1, \ldots, z_n$.  For
every integer $k\geq 3$ and parameters $\alpha \in \{1, 2\}$, $\beta
\in \{0, 1\}$, $\gamma \in \{0, 1\}$, $\delta \in \{0, 1, 2, 3\}$,
$\varepsilon \in \{0, 1\}$, $\zeta \in \{0, 1\}$ such that if $\alpha
= 2$ then $\varepsilon = \beta = \gamma$, let us build a graph $G_f(k,
\alpha, \beta, \gamma, \delta, \varepsilon, \zeta)$ with two
specified vertices $x,y$ of degree~2.  There will be a hole
containing $x$ and $y$ in $G_f(k, \alpha, \beta, \gamma, \delta,
\varepsilon, \zeta)$ if and only if there exists a truth assignment
satisfying $f$.  In $G_f(k, \alpha, \beta, \gamma, \delta,
\varepsilon, \zeta)$ (we will sometimes write $G_f$ for short), there
will be two kinds of edges: blue and red.  The reason for this
distinction will appear later.  Let us now describe $G_f$.

\subsubsection*{Pieces of $G_f$ arising from variables}

For each variable $z_i$ ($i=1, \ldots, n$), prepare a graph $G(z_i)$
with $4k$ vertices $a_{i,r},$ $b_{i, r},$ $a'_{i, r},$ $b'_{i, r}$,
$r\in\{1, \dots, k\}$ and $4(m+2)2k$ vertices $t_{i, 2pk+r},$ $f_{i,
2pk+r},$ $t'_{i, 2pk+r},$ $f'_{i, 2pk+r}$, $p \in \{0, \dots, m+1\}$,
$r \in \{0, \ldots, 2k-1\}$.  Add blue edges so that the four sets
$\{a_{i,1}, \dots a_{i,k},$ $t_{i, 0}, \ldots,$ $t_{i, 2k(m+2)-1},$
$b_{i, 1}, \dots, b_{i, k}\}$, $\{a_{i,1}, \dots a_{i,k},$ $f_{i, 0},
\ldots, f_{i, 2k(m+2)-1}, b_{i, 1}, \dots, b_{i, k}\}$, $\{a'_{i,1},
\dots a'_{i,k},$ $t'_{i, 0}, \ldots,$ $t'_{i, 2k(m+2)-1},$ $b'_{i, 1},
\dots, b'_{i, k}\}$, $\{a'_{i,1}, \dots a'_{i,k},$ $f'_{i, 0}, \ldots,
f'_{i, 2k(m+2)-1}, b'_{i, 1}, \dots, b'_{i, k}\}$ all induce paths
(and the vertices appear in this order along these paths).  See
Figure~\ref{fig.reco.1}.

\vspace{2ex}

\noindent Add red edges according to the value of $\alpha$, $\beta$,
$\gamma$, as follows:

\begin{itemize}
\item
  If $\alpha = 1$ then, for every $p=1, \dots, m+1$, add all edges
  between $\{t_{i, 2kp}, t_{i, 2kp+\beta}\}$ and $\{f_{i, 2kp}, f_{i,
    2kp+\gamma}\}$, between $\{f_{i, 2kp}, f_{i, 2kp+\gamma}\}$ and
  $\{t'_{i, 2kp}, t'_{i, 2kp+\beta}\}$, between $\{t'_{i, 2kp}, t'_{i,
    2kp+\beta}\}$ and $\{f'_{i, 2kp}, f'_{i, 2kp+\gamma}\}$, between
  $\{f'_{i, 2kp}, f'_{i, 2kp+\gamma}\}$ and $\{t_{i, 2kp}, t_{i,
    2kp+\beta}\}$.

\item
  If $\alpha = 2$ then, for every $p=1, \dots, m$, add all edges
  between $\{t_{i, 2kp + k-1}, t_{i, 2kp + k-1 + \beta}\}$ and
  $\{f_{i, 2kp + k-1}, f_{i, 2kp + k-1 + \gamma}\}$ ; for every $p=1,
  \dots, m+1$, add all edges between $\{f_{i, 2kp + k-1}, f_{i, 2kp +
    k-1 + \gamma}\}$ and $\{t'_{i, 2kp}, t'_{i, 2kp + \beta}\}$,
  between $\{t'_{i, 2kp}, t'_{i, 2kp+\beta}\}$ and $\{f'_{i, 2kp},
  f'_{i, 2kp+\gamma}\}$, between $\{f'_{i, 2kp}, f'_{i, 2kp+\gamma}\}$
  and $\{t_{i, 2k(p-1) + k-1}, t_{i, 2k(p-1) + k-1 + \beta}\}$.  

  See Figures~\ref{fig.reco.2}, \ref{fig.reco.3}.
\end{itemize}

\begin{figure}[p]
\center
\includegraphics{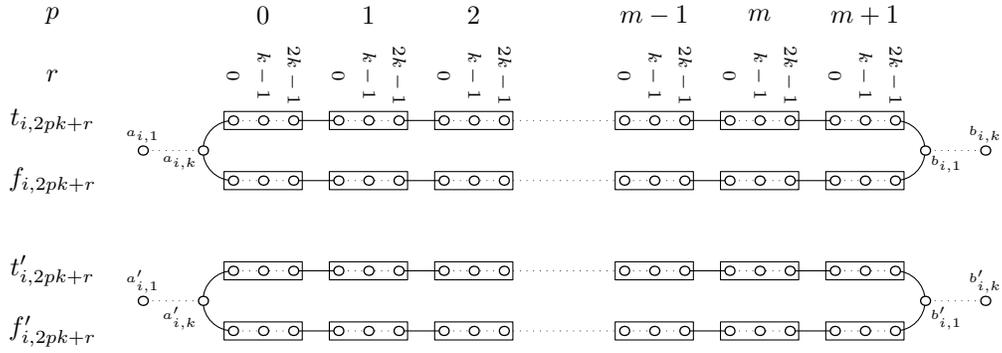}
\caption{The graph $G(z_i)$ (only blue edges are depicted)\label{fig.reco.1}}
\end{figure}

\begin{figure}[p]
\center
\includegraphics{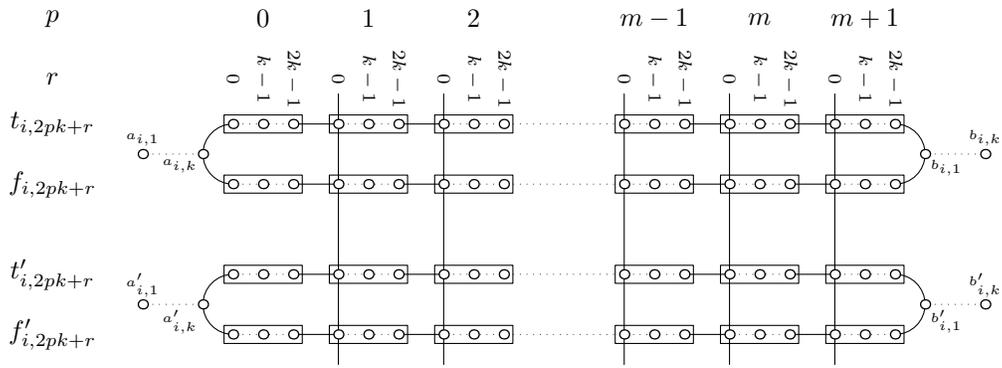}
\caption{The graph $G(z_i)$ when $\alpha = 1$, $\beta = 0$, $\gamma =
  0$\label{fig.reco.2}}
\end{figure}

\begin{figure}[p]
\center
\includegraphics{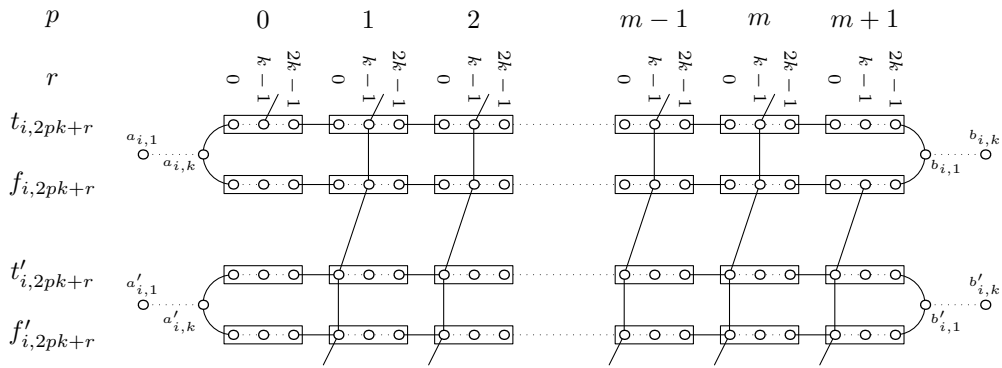}
\caption{The graph $G(z_i)$ when $\alpha = 2$, $\beta = 0$, $\gamma =
  0$\label{fig.reco.3}}
\end{figure}

\subsubsection*{Pieces of $G_f$ arising from clauses}

For each clause $C_j$ ($j=1, \ldots, m$), with $C_j = y_j^1\vee y_j^2
\vee y_j^3$, where each $y_j^q$ ($q=1, 2, 3$) is a literal from
$\{z_1,$ $\ldots,$ $z_n,$ $\overline{z}_1,$ $\ldots,$
$\overline{z}_n\}$, prepare a graph $G(C_j)$ with $2k$ vertices $c_{j,
  p}$, $d_{j, p}$, $p\in \{1, \dots, k\}$ and $6k$ vertices $u^q_{j,
  p}$, $q\in \{1, 2, 3\}$, $p\in \{1, \dots, 2k\}$.  Add blue edges so
that the three sets $\{c_{j, 1}, \dots, c_{j, k}, u^q_{j, 1}, \dots,
u^q_{j, 2k}, d_{j, 1}, \dots d_{j, k}\}$, $q\in \{1, 2, 3\}$ all
induce paths (and the vertices appear in this order along these
paths).

\vspace{2ex}

\noindent Add red edges according to the value of $\delta$:

\begin{itemize}
\item
  If
$\delta=0$, add no edge.  
\item
  If $\delta=1$, add $u^1_{j, 1}u^2_{j, 1}$,
$u^1_{j,2k}u^2_{j, 2k}$.  
\item
  If $\delta=2$, add $u^1_{j, 1}u^2_{j, 1}$, $u^1_{j,2k}u^2_{j, 2k}$,
  $u^1_{j, 1}u^3_{j, 1}$, $u^1_{j,2k}u^3_{j, 2k}$.  
\item
  If $\delta=3$, add $u^1_{j, 1}u^2_{j, 1}$, $u^1_{j,2k}u^2_{j, 2k}$,
  $u^1_{j, 1}u^3_{j, 1}$, $u^1_{j,2k}u^3_{j, 2k}$, $u^2_{j, 1}u^3_{j,
    1}$, $u^2_{j,2k}u^3_{j, 2k}$.  
\end{itemize}

\noindent See Figure~\ref{fig.reco.4}.

\begin{figure}[p]
\center
\includegraphics{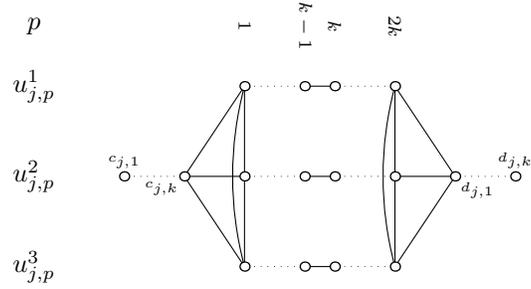}
\caption{The graph $G(C_j)$ when $\delta = 3$\label{fig.reco.4}}
\end{figure}

\begin{figure}[p]
\center
\includegraphics{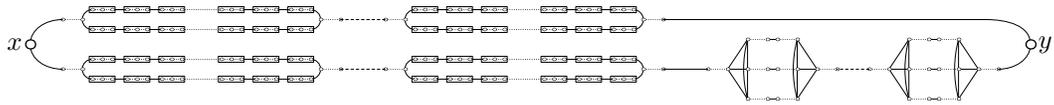}
\caption{The whole graph $G_f$\label{fig.reco.6}}
\end{figure}

\begin{figure}[p]
\center
\includegraphics{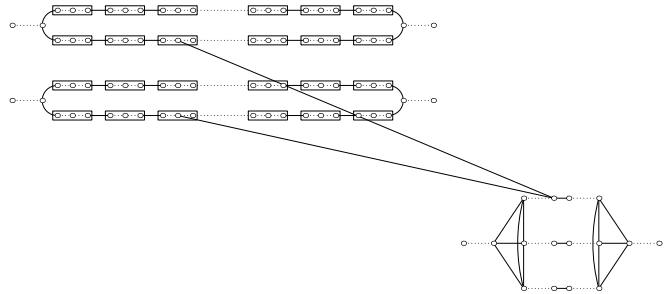}
\caption{Red edges between $G(z_i)$ and $G(C_j)$ when $\varepsilon =
  \zeta = 0$\label{fig.reco.5}}
\end{figure}

\subsubsection*{Gluing the pieces of $G_f$}

The graph $G_f(k, \alpha, \beta, \gamma, \delta, \varepsilon, \zeta)$
is obtained from the disjoint union of the $G(z_i)$'s and the
$G(C_j)$'s as follows.  For $i = 1, \ldots, n-1$, add blue edges
$b_{i, k}a_{i+1, 1}$ and $b'_{i, k}a'_{i+1, 1}$.  Add a blue edge
$b'_{n, k}c_{1, 1}$.  For $j=1, \ldots, m-1$, add a blue edge $d_{j,
k}c_{j+1, 1}$.  Introduce the two special vertices $x,y$ and add blue
edges $xa_{1, 1}, xa'_{1, 1}$ and $yd_{m, k}, yb_{n, k}$.  See
Figure~\ref{fig.reco.6}.

Add red edges according to $f, \varepsilon, \zeta$.  For $q=1, 2, 3$,
if $y_j^q = z_i$, then add all possible edges between $\{f_{i, 2kj + k
- 1}, f_{i, 2kj + k - 1 + \varepsilon} \}$ and $\{u_{j, k}^q, u_{j,
k+\zeta}^q\}$ and between $\{f'_{i, 2kj + k - 1}, f'_{i, 2kj + k - 1 +
\varepsilon} \}$ and $\{u_{j, k}^q, u_{j, k+\zeta}^q\}$;  if
$y_j^q=\overline{z}_i$ then add all possible edges between $\{t_{i,
2kj + k - 1}, t_{i, 2kj + k - 1 + \varepsilon} \}$ and $\{u_{j, k}^q,
u_{j, k+\zeta}^q\}$ and between $\{t'_{i, 2kj + k- 1}, t'_{i, 2kj + k
- 1 + \varepsilon} \}$ and $\{u_{j, k}^q, u_{j, k+\zeta}^q\}$.  See
Figure~\ref{fig.reco.5}.

Clearly the size of $G_f(k, \alpha, \beta, \gamma, \delta,
\varepsilon, \zeta)$ is polynomial (actually quadratic) in the size
$n+m$ of $f$, and $x,y$ are non-adjacent and both have degree two.

\begin{lemma}
  \label{l:nplemma}
$f$ is satisfied by a truth assignment if and only if $G_f(k, \alpha,
  \beta, \gamma, \delta, \varepsilon, \zeta)$ contains a hole passing
  through $x, y$.
\end{lemma}

\begin{proof}
Recall that if $\alpha = 2$ then $\varepsilon = \beta = \gamma$.  We
will prove the lemma for $\beta=0$, $\gamma=0$, $\varepsilon=0$,
$\zeta=0$ because the proof is essentially the same for the other
possible values.

Suppose that $f$ is satisfied by a truth assignment $\xi \in \{0,
1\}^n$.  We can build a hole in $G$ by selecting vertices as follows.
Select $x,y$.  For $i=1, \ldots, n$, select $a_{i, p}, b_{i, p},
a'_{i, p}, b'_{i, p}$ for all $p \in \{1, \dots, k\}$.  For $j=1,
\dots, m$, select $c_{j, p}, d_{j, p}$ for all $p\in \{1, \dots, k\}$.
If $\xi_i = 1$ select $t_{i, p}, t'_{i, p}$ for all $p \in \{0,
\ldots, 2k(m+2) - 1\}$.  If $\xi_i=0$ select $f_{i, p}, f'_{i, p}$ for
all $p \in \{0, \ldots, 2k(m+2)-1\}$.  For $j = 1, \ldots, m$, since
$\xi$ is a truth assignment satisfying $f$, at least one of the three
literals of $C_j$ is equal to $1$, say $y_j^q=1$ for some $q\in\{1, 2,
3\}$.  Then select $u_{j, p}^q$ for all $p\in \{1, \dots, 2k\}$.  Now
it is a routine matter to check that the selected vertices induce a
cycle $Z$ that contains $x,y$, and that $Z$ is chordless, so it is a
hole.  The main point is that there is no chord in $Z$ between some
subgraph $G(C_j)$ and some subgraph $G(z_i)$, for that would be either
an edge $t_{i, p}u_{j, r}^q$ with $y_j^q=z_i$ and $\xi_i=1$, or,
symmetrically, an edge $f_{i, p}u_{j, r}^q$ with $y_j^q =
\overline{z}_i$ and $\xi_i = 0$, and in either case this would
contradict the way the vertices of $Z$ were selected.

Conversely, suppose that $G_f(k, \alpha, \beta, \gamma, \delta,
\varepsilon, \zeta)$ admits a hole $Z$ that contains $x,y$.

\begin{claim}\label{clm:zgxi}
For $i=1, \ldots, n$, $Z$ contains at least $4k+4k(m+2)$ vertices of
$G(z_i)$: $4k$ of these are $a_{i, p}, a'_{i, p}, b_{i,p}, b'_{i, p}$
where $p\in \{1, \dots, k\}$, and the others are either the $t_{i, p},
t'_{i, p}$'s or the $f_{i, p}, f'_{i,p}$'s where $ p \in \{0, \ldots,
2k(m+2)-1\}$.
\end{claim}

\begin{proofclaim} 
Let us first deal with the case $i=1$.  Since $x\in Z$ has degree $2$,
$Z$ contains $a_{1, 1}, \dots, a_{1, k}$ and $a'_{1, 1}, \dots, a'_{1,
k}$.  Hence exactly one of $t_{1, 0}, f_{1, 0}$ is in $Z$.  Likewise
exactly one of $t'_{1, 0}, f'_{1, 0}$ is in $Z$.  If $t_{1, 0}, f'_{1,
0}$ are both in $Z$ then there is a contradiction: indeed, if
$\alpha=1$ then, $t_{1, 0}, \dots, t_{1, 2k}$ and $f'_{1, 0}, \dots,
f'_{1, 2k}$ must all be in $Z$, and since $t_{1, 2k}$ sees $f'_{1,
2k}$, $Z$ cannot go through $y$; and if $\alpha =2$ the proof is
similar.  Similarly, $t'_{1, 0}, f_{1, 0}$ cannot both be in $Z$.  So,
there exists a largest integer $p \leq 2k(m+2)-1$ such that either
$t_{1, 0}, \dots, t_{1, p}$ and $t'_{1, 0}, \dots, t'_{1, p}$ are all
in $Z$ or $f_{1, 0}, \dots, f_{1, p}$ and $f'_{1, 0}, \dots, f'_{1,
p}$ are all in $Z$.

We claim that $p = 2k(m+2) - 1$.  For otherwise, some vertex $w$ in
$\{t_{1, p}$, $t'_{1, p}$, $f_{1, p}$, $f'_{1, p}\}$ is incident to a
red edge $e$ of $Z$.  If $\alpha = 1$  then, up to a symmetry,
we assume that $t_{1, 0}, \dots, t_{1, p}$ and $t'_{1, 0}, \dots,
t'_{1, p}$ are all in $Z$.  Let $w'$ be the vertex of $e$ that is not
$w$.  Then $w'$ (which is either an $f_{1, \cdot}$, an $f'_{1, \cdot}$
or a $u^\cdot_{j, \cdot}$) is a neighbor of both $t_{1, p}$, $t'_{1,
p}$.  Hence, $Z$ cannot go through $y$, a contradiction.  This proves
our claim when $\alpha = 1$.  If $\alpha=2$, we distinguish
between the following six cases.  \\
{\it Case~1: $p = k-1$.} Then $e = t_{1, k-1} f'_{1, 2k}$.  Clearly
$t_{1, 0}, \dots, t_{1, k-1}$ must all be in $Z$.  If $t'_{1, 0},
\dots, t'_{1, 2k}$ are in $Z$, there is a contradiction because of
$t'_{1, 2k} f'_{1,2k}$, and if $f'_{1, 0}, \dots, f'_{1, 2k}$ are in
$Z$, there is a contradiction because of $e$.  \\
{\it Case~2: $p = 2kl$ where $1\leq l \leq m+1$ and $w = t'_{1,
2kl}$.} Then $e$ is $t'_{1, 2kl} f_{1, 2kl + k -1}$ or $t'_{1, 2kl}
f'_{1, 2kl}$.  In either case, $t_{1, 2kl}, \dots, t_{1, 2kl+k-1}$ are
all in $Z$, and there is a contradiction because of the red edge
$f_{1, 2kl + k - 1} t_{1, 2kl + k -1}$ or $t_{1, 2(l-1)k + k - 1}
f'_{1, 2kl}$, or when $l=m+1$ because of $b_{1, 1}$.  \\
{\it Case~3: $p = 2kl$ where $1\leq l \leq m+1$ and $w = f'_{1,
2kl}$.} Then $e$ is $f'_{1, 2kl} t_{1, 2(l-1)k + k - 1}$ or $t'_{1,
2kl} f'_{1, 2kl}$.  In either case, $f_{1, 2kl}, \dots, f_{1,
2kl+k-1}$ are all in $Z$, and there is a contradiction because of the
red edge $t_{1, 2(l-1)k + k - 1} f_{1, 2(l-1)k + k -1}$ or $t'_{1,
2kl} f_{1, 2kl+k-1}$, or when $l=1$ because of $a_{1, k}$.  \\
{\it Case~4: $p = 2kl + k - 1$ where $1 \leq l \leq m$ and $w = t_{1,
2kl + k - 1}$.} Then $e$ is $t_{1, 2kl+k-1} f_{1, 2kl + k - 1}$,
$t_{1, 2kl + k - 1} f'_{1, 2(l+1)k}$, or $t_{1, 2kl+k-1} u_{j, k}^q$
for some $j, q$.  In the last case, there is a contradiction since
$t'_{1, 2kl + k - 1}\in Z$ also sees $u_{j,k}^q$.  For the same
reason, $t'_{1, 2kl + k - 1} u_{j,k}^q$ is not an edge of $Z$ and
$t'_{1, 2kl+k-1}, \dots, t'_{1, 2(l+1)k}$ are all in $Z$.  So there is
a contradiction because of the red edge $t'_{1, 2kl} f_{1, 2kl+k-1}$
or $t'_{1, 2(l+1)k} f'_{1, 2(l+1)k}$.  \\
{\it Case~5: $p = 2kl + k - 1$ where $2 \leq l \leq m$ and $w = f_{1,
2kl + k - 1}$.} Then $e$ is either $f_{1, 2kl+k-1} t_{1, 2kl + k - 1}$ or
$f_{1, 2kl + k - 1} t'_{1, 2kl}$ or $f_{1, 2kl+k-1} u_{j, k}^q$ for
some $j, q$.  In the last case, there is a contradiction since $f'_{1,
2kl + k - 1}\in Z$ also sees $u_{j,k}^q$.  For the same reason,
$f'_{1, 2kl + k - 1} u_{j,k}^q$ is not an edge of $Z$ and $f'_{1,
2kl+k-1}, \dots, f'_{1, 2(l+1)k}$ are all in $Z$.  So there is a
contradiction because of the red edge $t'_{1, 2kl} f'_{1, 2kl}$ or
$t_{1, 2kl+k-1} f'_{1, 2(l+1)}$.  \\
{\it Case~6: $p=2k(m+1) + k -1$ and $w = f_{1, 2k(m+1) + k -1}$.} Then
there is a contradiction because of the red edge $t'_{1,
2k(m+1)}f'_{1, 2k(m+1)}$.  This proves our claim.

Since $p = 2k(m+2) - 1$, $b_{1, 1}$ is in $Z$.  We claim that $b_{1,
  2}$ is in $Z$.  For otherwise, the two neighbors of $b_{1, 1}$ in
$Z$ are $t_{1, 2k(m+2) - 1}$ and $f_{1, 2k(m+2) - 1}$.  This is a
contradiction because of the red edges $t_{1, 2km + k - 1} f'_{1,
  2k(m+1)}$, $t'_{1, 2k(m+1)} f_{1, 2k(m+1) + k - 1}$ (if $\alpha =
2$) or $t_{1, 2k(m+1)} f'_{1, 2k(m+1)}$, $t'_{1, 2k(m+1)} f_{1,
  2k(m+1)}$ (if $\alpha = 1$).  Similarly, $b'_{1, 1}, b'_{1, 2}$ are
in $Z$.  So $b_{1, 1}, \dots, b_{1, k}$ and $b'_{1, 1}, \dots, b'_{1,
  k}$ are all in $Z$.

This proves~(\ref{clm:zgxi}) for $i=1$.  The proof for $i=2, ..., n$ is
essentially the same as for $i=1$.
\end{proofclaim}

\begin{claim}\label{clm:zgcj}
For $j=1, \ldots, m$, $Z$ contains $c_{j, 1}, \dots, c_{j, k}$, $d_{j,
1}, \dots, d_{j, k}$ and exactly one of $\{u_{j, 1}^1, \dots, u_{j,
2k}^1\}$, $\{u_{j, 1}^2, \dots, u_{j, 2k}^2\}$, $\{u_{j, 1}^3, \dots,
u_{j, 2k}^3\}$.
\end{claim}

\begin{proofclaim} 
Let us first deal with the case $j=1$.  By~(\ref{clm:zgxi}), $b'_{n,
  k}$ is in $Z$ and so $c_{1, 1}, \dots, c_{1, k}$ are all in $Z$.
Consequently exactly one of $u_{1, 1}^1, u_{1, 1}^2, u_{1, 1}^3$ is in
$Z$, say $u_{1, 1}^1$ up to a symmetry.  Note that the neighbour of
$u_1^1$ in $Z\setminus c_{1, k}$ cannot be a vertex among $u_{1, 1}^2,
u_{1, 1}^3$ for this would imply that $Z$ contains a triangle.  Hence
$u_{1, 2}^1, \dots, u^1_{1, k}$ are all in $Z$.  The neighbour of
$u_{1, k}^1$ in $Z\setminus u^1_{1, k-1}$ cannot be in some $G(z_i)$
($1\le i\le n$).  Else, up to a symmetry we assume that this neighbor
is $t_{1, p}$, $p\in\{0, \dots, 2k(m+2)-1\}$.  If $t_{1, p}\in Z$,
there is a contradiction because then $t'_{1, p}$ is also in $Z$
by~(\ref{clm:zgxi}) and $t'_{1, p}$ would be a third neighbour of
$u_{1, k}^1$ in $Z$.  If $t_{1, p}\notin Z$, there is a contradiction
because then the neighbor of $t_{1, p}$ in $Z\setminus u_{1, k}^1$
must be $t_{1, p+1}$ (or symmetrically $t_{1, p-1}$) for otherwise $Z$
contains a triangle.  So, $t_{1, p+1}, t_{1, p+2}, \dots$ must be in
$Z$, till reaching a vertex having a neighbor $f_{1, p'}$ or $f'_{1,
  p'}$ in $Z$ (whatever $\alpha$).  Thus the neighbour of $u_{1, k}^1$
in $Z\setminus u_{1, k-1}^1$ is $u_{1, k+1}^1$.  Similarly, we prove
that $u_{1, k+2}, \dots, u_{1, 2k}$ are in $Z$, that $d_{1, 1}, \dots,
d_{1, k}$ are in $Z$, and so the claim holds for $j=1$.  The proof of
the claim for $j=2, ..., m$ is essentially the same.
\end{proofclaim}

Together with $x, y$, the vertices of $Z$ found in~(\ref{clm:zgxi})
and~(\ref{clm:zgcj}) actually induce a cycle.  So, since $Z$ is a
hole, they are the members of $Z$ and we can replace ``at least'' by
``exactly'' in~(\ref{clm:zgxi}).  We can now make a Boolean vector
$\xi$ as follows.  For $i=1, \ldots, n$, if $Z$ contains $t_{i, 0},
t'_{i, 0}$ set $\xi_i = 1$; if $Z$ contains $f_{i, 0}, f'_{i, 0}$ set
$\xi_i = 0$.  By~(\ref{clm:zgxi}) this is consistent.  Consider any
clause $C_j$ ($1\le j\le m$).  By~(\ref{clm:zgcj}) and up to symmetry
we may assume that $u_{j, k}^1$ is in $Z$.  If $y_j^1 = z_i$ for some
$i\in\{1, .., n\}$, then the construction of $G$ implies that $f_{i,
2kj+k-1}, f'_{i, 2j+k-1}$ are not in $Z$, so $t_{i, 2kj+k-1}, t'_{i,
2kj+k-1}$ are in $Z$, so $\xi_i=1$, so clause $C_j$ is satisfied by
$x_i$.  If $y_j^1 = \overline{z}_i$ for some $i\in\{1, \ldots, n\}$,
then the construction of $G_f$ implies that $t_{i, 2kj+k-1}, t'_{i,
2kj+k-1}$ are not in $Z$, so $f_{i, 2kj+k-1}, f'_{i, 2kj+k-1}$ are in
$Z$, so $\xi_i=0$, so clause $C_j$ is satisfied by $\overline{z}_i$.
Thus $\xi$ is a truth assignment satisfying $f$.
\end{proof}

\begin{theorem}
  \label{th:bienstock}
  Let $k\geq 5$ be an integer.  Then $\Gamma_{\{C_3, \dots, C_k, K_{1,
      6}\}}$ and $\Gamma_{\{I_1, \dots, I_k, C_5, \dots, C_k, K_{1,
      4}\}}$ are NP-complete.
 \end{theorem}
\begin{proof}
It is a routine matter to check that the graph $G_f(k, 2, 0, 0, 0, 0,
0)$ contains no $C_{l}$ ($3 \leq l \leq k$) and no $K_{1,6}$ (in fact
it has no vertex of degree at least $6$).  So Lemma~\ref{l:nplemma}
implies that $\Gamma_{\{C_3, \dots, C_k, K_{1, 6}\}}$ is NP-complete.

It is a routine matter to check that the graph $G_f(k, 1, 1, 1, 3, 1,
1)$ contains no $K_{1, 4}$, no $C_{l}$ ($5 \leq l \leq k$) and no
$I_{l'}$ ($1 \leq l' \leq k$).  So Lemma~\ref{l:nplemma} implies that
$\Gamma_{\{K_{1, 4}, C_5, \dots, C_k, I_1, \dots, I_k\}}$ is
NP-complete.
\end{proof}

\subsection{Complexity of $\Gamma_{\{H\}}$ when $H$ is a connected graph}

\begin{theorem}
  \label{th:class}
Let $H$ be a connected graph.  Then one of the following holds:
  \begin{itemize}
    \item 
$H$ is a path or a subdivided claw and $\Gamma_{\{H\}}$ is polynomial.
    \item 
$H$ contains one of $K_{1, 4}$, $I_k$ for some $k\geq 1$, or $C_l$ for
some $l\geq 3$ as an induced subgraph and $\Gamma_{\{H\}}$ is
NP-complete.
  \end{itemize}
\end{theorem}
\begin{proof}
If $H$ contains one of $K_{1, 4}$, $I_k$ for some $k\geq 1$, or $C_l$
for some $l\geq 3$ as an induced subgraph then $\Gamma_{\{H\}}$ is
NP-complete by Theorem~\ref{th:bienstock}.  Else, $H$ is a tree since
it contains no $C_l$, $l\geq 3$.  If $H$ has no vertex of degree at
least~$3$, then $H$ is a path and $\Gamma_{\{H\}}$ is polynomial by
Theorem~\ref{th:clawpoly}.  If $H$ has a single vertex of degree at
least~$3$, then this vertex has degree 3 because $H$ contains no
$K_{1, 4}$.  So, $H$ is a subdivided claw and $\Gamma_{\{H\}}$ is
polynomial by Theorem~\ref{th:clawpoly}.  If $H$ has at least two
vertices of degree at least~3 then $H$ contains an $I_l$, where $l$ is
the minimum length of a path of $H$ joining two such vertices.  This
is a contradiction.
\end{proof}

Interestingly, the following analogous result for finding maximum
stable sets in $H$-free graphs was proved by Alekseev:

\begin{theorem}[Alekseev, \cite{alekseev:83}]
  \label{th:alekseev}
Let $H$ be a connected graph that is not a path nor a subdivided claw.
Then the problem of finding a maximum stable set in $H$-free graphs is
NP-hard.
\end{theorem}

But the complexity of the maximum stable set problem is not known
in general for $H$-free graphs when $H$ is a path or a subdivided
claw. See~\cite{hertz.lozin:survey} for a survey.

\subsection{NP-complete cases (bounded degree)}

Here, we will show that $\Gamma^d$ is NP-complete when $d \ge 3$ and
polynomial when $d = 2$.  If $\cal S$ is any finite list of cycles
$C_{k_1}, C_{k_2}, \ldots, C_{k_m}$, then we will also show that
$\Gamma^3_{\cal S}$ is NP-complete as long as $C_6 \notin \cal S$.

Let $f$ be an instance of {\sc $3$-Satisfiability}, consisting of $m$
clauses $C_1, \ldots, C_m$ on $n$ variables $z_1, \ldots, z_n$.  For
each clause $C_j$ ($j=1, \ldots, m$), with $C_j = y_{3j-2} \vee
y_{3j-1} \vee y_{3j}$, then $y_i$ ($i = 1, \ldots, 3m$) is a literal
from $\{z_1,$ $\ldots,$ $z_n,$ $\overline{z}_1,$ $\ldots,$
$\overline{z}_n\}$.

Let us build a graph $G_f$ with two specified vertices $x$ and $y$
of degree 2 such that $\Delta(G_f) = 3$ .  There will be a hole
containing $x$ and $y$ in $G_f$ if and only if there exists a truth
assignment satisfying $f$.

For each literal $y_j$ ($j=1, \ldots, 3m$), prepare a graph $G(y_j)$
on $20$ vertices $\alpha$, $\alpha'$, $\alpha^{1+}$, $\ldots$,
$\alpha^{4+}$, $\alpha^{1-}$, $\ldots$, $\alpha^{4-}$, $\beta$,
$\beta'$, $\beta^{1+}$, $\ldots$, $\beta^{4+}$, $\beta^{1-}$,
$\ldots$, $\beta^{4-}$.  (We drop the subscript $j$ in the labels of
the vertices for clarity).

For $i = 1,2,3$ add the edges $\alpha^{i+}\alpha^{(i+1)+}$,
$\beta^{i+}\beta^{(i+1)+}$, $\alpha^{i-}\alpha^{(i+1)-}$,
$\beta^{i-}\beta^{(i+1)-}$.  Also add the edges
$\alpha^{1+}\beta^{1-}$, $\alpha^{1-}\beta^{1+}$,
$\alpha^{4+}\beta^{4-}$, $\alpha^{4-}\beta^{4+}$, $\alpha\alpha^{1+}$,
$\alpha\alpha^{1-}$, $\alpha^{4+}\alpha'$, $\alpha^{4-}\alpha'$,
$\beta\beta^{1+}$, $\beta\beta^{1-}$, $\beta^{4+}\beta'$,
$\beta^{4-}\beta'$.  See Figure~\ref{fig:Gyj}.

\begin{figure}
	\centering
                \includegraphics{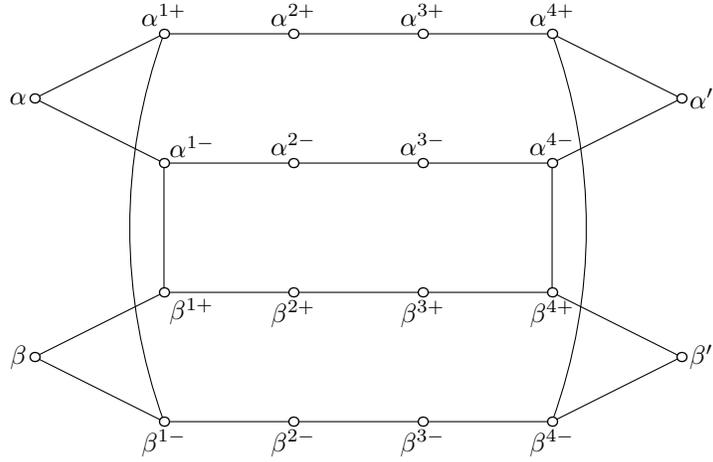}
		\caption{The graph $G(y_j)$\label{fig:Gyj}}
\end{figure}


For each clause $C_j$ $(j = 1, \ldots, m)$, prepare a graph $G(C_j)$
with 10 vertices $c^{1+}$, $c^{2+}$, $c^{3+}$, $c^{1-}$, $c^{2-}$,
$c^{3-}$, $c^{0+}$, $c^{12+}$, $c^{0-}$, $c^{12-}$.  (We drop the
subscript $j$ in the labels of the vertices for clarity).  

Add the edges $c^{12+}c^{1+}$, $c^{12+}c^{2+}$, $c^{12-}c^{1-}$,
$c^{12-}c^{2-}$, $c^{0+}c^{12+}$, $c^{0+}c^{3+}$, $c^{0-}c^{12-}$,
$c^{0-}c^{3-}$.  See Figure~\ref{fig:GCj}.

\begin{figure}
	\centering
        \includegraphics{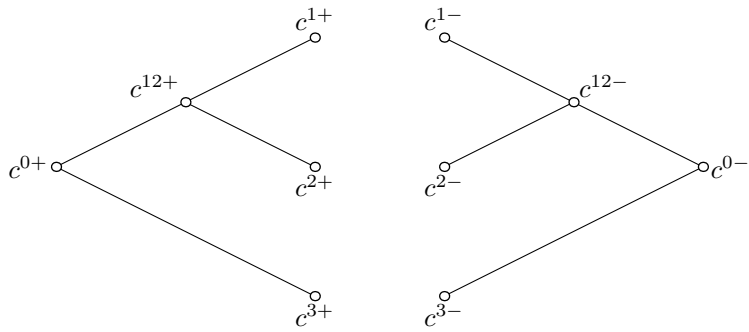}
		\caption{The graph $G(C_j)$\label{fig:GCj}}
\end{figure}



For each variable $z_i$ $(i = 1, \ldots, n)$, prepare a graph $G(z_i)$
with $2z^-_i + 2z^+_i$ vertices, where $z^-_i$ is the number of times
$\overline{z}_i$ appears in clauses $C_1, \ldots, C_m$ and $z^+_i$ is
the number of times $z_i$ appears in clauses $C_1, \ldots, C_m$.

Let $G(z_i)$ consist of two internally disjoint paths $P^+_i$ and
$P^-_i$ with common endpoints $d^+_i$ and $d^-_i$ and lengths $1 +
2z^-_i$ and $1 + 2z^+_i$ respectively. Label the vertices of $P^+_i$
as $d^+_i$, $p^+_{i,1}$, $\ldots$, $p^+_{i,2f_i}$, $d^-_i$ and label
the vertices of $P^-_i$ as $d^+_i$, $p^-_{i,1}$, $\ldots$,
$p^-_{i,2g_i}$, $d^-_i$. See Figure~\ref{fig:Gzi}.

\begin{figure}
	\centering
        \includegraphics{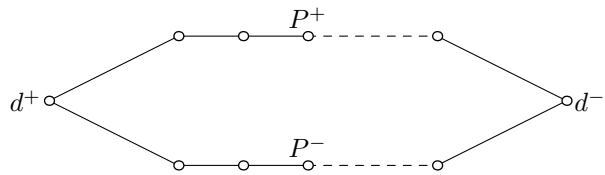}
		\caption{The graph $G(z_i)$\label{fig:Gzi}}
\end{figure}

\begin{figure}
  \label{fig:graphfinal}
	\centering
        \includegraphics{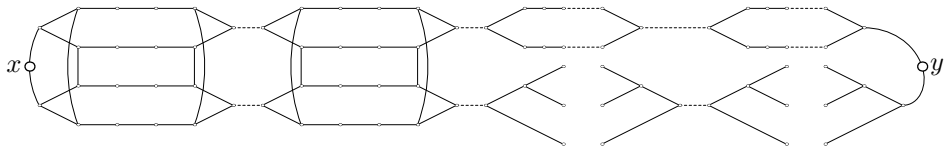}
		\caption{The final graph $G_f$\label{fig:graphfinalLin}}
\end{figure}


The final graph $G_f$ (see Figure~\ref{fig:graphfinalLin}) will be
constructed from the disjoint union of all the graphs $G(y_j)$,
$G(C_i)$, and $G(z_i)$ with the following modifications:

\begin{itemize}
\item 
  For $j = 1, \ldots, 3m-1$, add the edges
  $\alpha'_j\alpha_{j+1}$ and $\beta'_j\beta_{j+1}$.   
\item 
  For $j = 1, \ldots, m-1$, add the edge $c^{0-}_j
  c^{0+}_{j+1}$.  
\item 
  For $i = 1, \ldots, n-1$, add the edge $d^-_i d^+_{i+1}$.
\item 
  For $i = 1, \ldots, n$, let $y_{n_1}$, $\ldots$, $y_{n_{z^-_i}}$ be
  the occurrences of $\overline{z}_i$ over all literals.  For $j = 1,
  \ldots, z^-_i$, delete the edge $p^+_{i,2j-1}p^+_{i,2j}$ and add the
  four edges $p^+_{i,2j-1}\alpha^{2+}_{n_j}$,
  $p^+_{i,2j-1}\beta^{2+}_{n_j}$, $p^+_{i,2j}\alpha^{3+}_{n_j}$,
  $p^+_{i,2j}\beta^{3+}_{n_j}$.
\item 
  For $i = 1, \ldots, n$, let $y_{n_1}$, $\ldots$, $y_{n_{z^+_i}}$ be
  the occurrences of $z_i$ over all literals.  For $j = 1, 2, \ldots,
  z^+_i$, delete the edge $p^-_{i,2j-1}p^-_{i,2j}$ and add the four
  edges $p^-_{i,2j-1}\alpha^{2+}_{n_j}$,
  $p^-_{i,2j-1}\beta^{2+}_{n_j}$, $p^-_{i,2j}\alpha^{3+}_{n_j}$,
  $p^-_{i,2j}\beta^{3+}_{n_j}$.
\item 
  For $i = 1, \ldots, m$ and $j = 1, 2, 3$, add the edges
  $\alpha^{2-}_{3(i-1)+j}c^{j+}_i$, $\alpha^{3-}_{3(i-1)+j}c^{j-}_i$,
  $\beta^{2-}_{3(i-1)+j}c^{j+}_i$, $\beta^{3-}_{3(i-1)+j}c^{j-}_i$.
\item 
  Add the edges $\alpha'_{3m}d^+_1$ and $\beta'_{3m}c^{0+}_1$
\item 
  Add the vertex $x$ and add the edges $x\alpha_1$ and
  $x\beta_1$.  \item Add the vertex $y$ and add the edges $y c^{0-}_m$
  and $yd^-_n$.
\end{itemize}
It is easy to verify that $\Delta(G_f) = 3$, that the size of $G_f$ is
polynomial (actually linear) in the size $n + m$ of $f$, and that
$x,y$ are non-adjacent and both have degree two.


\begin{lemma}
\label{NPmax3lemma}
$f$ is satisfied by a truth assignment if and only if $G_f$ contains a
hole passing through $x$ and $y$.
\end{lemma}

\begin{proof}
First assume that $f$ is satisfied by a truth assignment $\xi \in
\{0,1\}^n$. We will pick a set of vertices that induce a hole
containing $x$ and $y$.

\begin{enumerate}
	\item 
Pick vertices $x$ and $y$.
	\item 
For $i = 1, \ldots, 3m$, pick the vertices $\alpha_i$, $\alpha'_i$,
$\beta_i$, $\beta'_i$.
	\item 
For $i = 1, \ldots, 3m$, if $y_i$ is satisfied by $\xi$, then pick the
vertices $\alpha^{1+}_i$, $\alpha^{2+}_i$, $\alpha^{3+}_i$,
$\alpha^{4+}_i$, $\beta^{1+}_i$, $\beta^{2+}_i$, $\beta^{3+}_i$, and
$\beta^{4+}_i$.  Otherwise, pick the vertices $\alpha^{1-}_i$,
$\alpha^{2-}_i$, $\alpha^{3-}_i$, $\alpha^{4-}_i$, $\beta^{1-}_i$,
$\beta^{2-}_i$, $\beta^{3-}_i$, and $\beta^{4-}_i$.
	\item 
For $i = 1, \ldots, n$, if $\xi_i = 1$, then pick all the vertices of
the path $P^+_i$ and all the neighbors of the vertices in $P^+_i$ of
the form $\alpha^{2+}_k$ or $\alpha^{3+}_k$ for any $k$.
	\item 
For $i = 1, \ldots, n$, if $\xi_i = 0$, then pick all the vertices of
the path $P^-_i$ and all the neighbors of the vertices in $P^-_i$ of
the form $\alpha^{2+}_k$ or $\alpha^{3+}_k$ for any $k$.
	\item 
For $i = 1, \ldots, m$, pick the vertices $c^{0+}_i$ and $c^{0-}_i$.
Choose any $j \in \{3i-2, 3i-1, 3i\}$ such that $\xi$ satisfies $y_j$.
Pick vertices $\alpha^{2-}_j$, and $\alpha^{3-}_j$.  If $j = 3i-2$,
then pick the vertices $c^{12+}_i$, $c^{1+}_i$, $c^{1-}_i$,
$c^{12-}_i$.  If $j = 3i-1$, then pick the vertices $c^{12+}_i$,
$c^{2+}_i$, $c^{2-}_i$, $c^{12-}_i$.  If $j = 3i$, then pick the
vertices $c^{3+}_i$ and $c^{3-}_i$.
\end{enumerate}
	
It suffices to show that the chosen vertices induce a hole containing
$x$ and $y$.  The only potential problem is that for some $k$, one of
the vertices $\alpha^{2+}_k$, $\alpha^{3+}_k$, $\alpha^{2-}_k$, or
$\alpha^{3-}_k$ was chosen more than once.  If $\alpha^{2+}_k$ and
$\alpha^{3+}_k$ were picked in Step 3, then $y_k$ is satisfied by
$\xi$.  Therefore, $\alpha^{2+}_k$ and $\alpha^{3+}_k$ were not chosen
in Step 4 or Step 5.  Similarly, if $\alpha^{2-}_k$ and
$\alpha^{3-}_k$ were picked in Step 6, then $y_k$ is satisfied by
$\xi$ and $\alpha^{2-}_k$ and $\alpha^{3-}_k$ were not picked in Step
3.  Thus, the chosen vertices induce a hole in $G$ containing vertices
$x$ and $y$.

Now assume $G_f$ contains a hole $H$ passing through $x$ and $y$.  The
hole $H$ must contain $\alpha_1$ and $\beta_1$ since they are the only
two neighbors of $x$.  Next, either both $\alpha_1^{1+}$ and
$\beta_1^{1+}$ are in $H$, or both $\alpha_1^{1-}$ and $\beta_1^{1-}$
are in $H$.

Without loss of generality, let $\alpha_1^{1+}$ and $\beta_1^{1+}$ be
in $H$ (the same reasoning that follows will hold true for the other
case).  Since $\beta_1^{1-}$ and $\alpha_1^{1-}$ are both neighbors of
two members in $H$, they cannot be in $H$.  Thus, $\alpha^{2+}_1$ and
$\beta^{2+}_1$ must be in $H$.  Since $\alpha^{2+}_1$ and
$\beta^{2+}_1$ have the same neighbor outside $G(y_1)$, it follows
that $H$ must contain $\alpha^{3+}_1$ and $\beta^{3+}_1$.  Also, $H$
must contain $\alpha^{4+}_1$ and $\beta^{4+}_1$.  Suppose that
$\alpha^{4-}_1$ and $\beta^{4-}_1$ are in $H$.  Because
$\alpha^{i-}_1$ has the same neighbor as $\beta^{i-}_1$ outside
$G(y_1)$ for $i = 2, 3$, it follows that $H$ must contain
$\alpha^{3-}_1$, $\alpha^{2-}_1$, and $\alpha_1^{1-}$.  But then $H$
is not a hole containing $b$, a contradiction.  Therefore,
$\alpha^{4-}_1$ and $\beta^{4-}_1$ cannot both be in $H$, so $H$ must
contain $\alpha'_1$, $\beta'_1$, $\alpha_2$, and $\beta_2$.

By induction, we see for $i = 1, 2, \ldots, 3m$ that $H$ must contain
$\alpha_i$, $\alpha'_i$, $\beta_i$, $\beta'_i$.  Also, for each $i$,
either $H$ contains $\alpha^{1+}_i$, $\alpha^{2+}_i$, $\alpha^{3+}_i$,
$\alpha^{4+}_i$, $\beta^{1+}_i$, $\beta^{2+}_i$, $\beta^{3+}_i$,
$\beta^{4+}_i$ or $H$ contains $\alpha^{1-}_i$, $\alpha^{2-}_i$,
$\alpha^{3-}_i$, $\alpha^{4-}_i$, $\beta^{1-}_i$, $\beta^{2-}_i$,
$\beta^{3-}_i$, $\beta^{4-}_i$.

As a result, $H$ must also contain $d^+_1$ and $c^{0+}_1$.  By
symmetry, we may assume $H$ contains $p^+_{1,1}$ and $\alpha^{2+}_k$
for some $k$.  Since $\alpha_k^{1+}$ is adjacent to two vertices in
$H$, $H$ must contain $\alpha^{3+}_k$.  Similarly, $H$ cannot contain
$\alpha_k^{4+}$, so $H$ contains $p^+_{1,2}$ and $p^+_{1,3}$.  By
induction, we see that $H$ contains $p^+_{1,i}$ for $i = 1, 2, \ldots,
z^+_i$ and $d^-_1$.  If $H$ contains $p^-_{1,z^-_i}$, then $H$ must
contain $p^-_{1,i}$ for $i = z^-_i, \ldots, 1$, a contradiction.
Thus, $H$ must contain $d^+_2$.  By induction, for $i = 1, 2, \ldots,
n$, we see that $H$ contains all the vertices of the path $P^+_i$ or
$P^-_i$ and by symmetry, we may assume $H$ contains all the neighbors
of the vertices in $P^+_i$ or $P^-_i$ of the form $\alpha^{2+}_k$ or
$\alpha^{3+}_k$ for any $k$.

Similarly, for $i = 1, 2, \ldots, m$, it follows that $H$ must contain
$c^{0+}_i$ and $c^{0-}_i$.  Also, $H$ contains one of the following:
\begin{itemize}
	\item $c^{12+}_i$, $c^{1+}_i$, $c^{1-}_i$, $c^{12-}_i$ and
          either $\alpha^{2-}_j$ and $\alpha^{3-}_j$ or $\beta^{2-}_j$
          and $\beta^{3-}_j$ (where $\alpha^{2-}_j$ is adjacent to
          $c^{1+}_i$).
	\item $c^{12+}_i$, $c^{2+}_i$, $c^{2-}_i$, $c^{12-}_i$ and
          either $\alpha^{2-}_j$ and $\alpha^{3-}_j$ or $\beta^{2-}_j$
          and $\beta^{3-}_j$ (where $\alpha^{2-}_j$ is adjacent to
          $c^{2+}_i$).
	\item $c^{3+}_i$ and $c^{3-}_i$ and either $\alpha^{2-}_j$ and
          $\alpha^{3-}_j$ or $\beta^{2-}_j$ and $\beta^{3-}_j$ (where
          $\alpha^{2-}_j$ is adjacent to $c^{3+}_i$).
\end{itemize}

We can recover the satisfying assignment $\xi$ as follows.  For $i =
1, 2, \ldots, n$, set $\xi_i = 1$ if the vertices of $P^+_i$ are in
$H$ and set $\xi_i = 0$ if the vertices of $P^-_i$ are in $H$.  By
construction, it is easy to verify that at least one literal in every
clause is satisfied, so $\xi$ is indeed a satisfying assignment.
\end{proof}

Note that the graph $G_f$ used above contains several $C_6$'s that we
could not eliminate, induced for instance by $\alpha, \alpha^{1+},
\beta^{1-}, \beta, \beta^{1+}, \alpha^{1-}$.

\newpage
\begin{theorem}
\label{th:delta3}
The following statements hold:
\begin{itemize}
\item 
For any $d \in \mathbb{Z}$ with $d \ge 2$, the problem $\Gamma^d$ is
NP-complete when $d \ge 3$ and polynomial when $d = 2$.  
\item 
If $\mathcal{H}$ is any finite list of cycles $C_{k_1}, C_{k_2},
\ldots, C_{k_m}$ such that $C_6 \notin \mathcal{H}$, then
$\Gamma^3_\mathcal{H}$ is NP-complete.
\end{itemize}
\end{theorem}

\begin{proof}
In the above reduction, $\Delta(G_f) = 3$ so $\Gamma^d$ is NP-complete
for $d \ge 3$.  When $d = 2$, there is a simple $O(n)$ algorithm.  Any
hole containing $x$ and $y$ must be a component of $G$ so pick the
vertex $x$ and consider the component $C$ of $G$ that contains $x$.
It takes $O(n)$ time to verify whether $C$ is a hole containing $x$
and $y$ or not.

To show the second statement, let $K$ be the length of the longest
cycle in $\mathcal{H}$.  In the above reduction, do the following
modifications.

\begin{itemize}
	\item 
	For $i = 1, 2, 3$ and $j = 1,2,\ldots, 3m$, replace the edges
	$\alpha^{i+}_j\alpha^{(i+1)+}_j$,
	$\alpha^{i-}_j\alpha^{(i+1)-}_j$,
	$\beta^{i+}_j\beta^{(i+1)+}_j$, and
	$\beta^{i-}_j\beta^{(i+1)-}_j$ by paths of length $K$.  
	\item
	For $j = 1,2,\ldots, 3m-1$, replace the edges
	$\alpha'_j\alpha_{j+1}$ and $\beta'_j\beta_{j+1}$ by paths of
	length $K$.  
	\item 
	Replace the edges $x\alpha_1$ and $x\beta_1$ by paths of
	length $K$.
\end{itemize}

This new reduction is polynomial in $n, m$ and contains no graph of
the list $\mathcal{H}$.  The proof of Lemma~\ref{NPmax3lemma} still
holds for this new reduction therefore $\Gamma^3_\mathcal{H}$ is
NP-complete.
\end{proof}

\section{$\Pi_B$ for some special s-graphs}

\subsection{Holes with pending edges and trees}
\label{sec:antenna}

Here, we study $\Pi_{B_4}$, \dots, $\Pi_{B_7}$ where $B_4,\dots, B_7$
are the s-graphs depicted on Figure~\ref{fig:antenna}.  Our motivation
is simply to give a striking example and to point out that
surprisingly, pending edges of s-graphs matter and that even an
s-graph with no cycle can lead to NP-complete problems.

\begin{theorem}
  \label{th:antena}
  There is an $O(n^{13})$-time algorithm for $\Pi_{B_4}$ but
  $\Pi_{B_5}$ is NP-complete.
\end{theorem}

\begin{proof}
A realisation of $B_4$ has exactly one vertex of degree 3 and one
vertex of degree 4.  Let us say that the realisation $H$ is
\emph{short} if the distance between these two vertices in $H$ is at
most~3.  Detecting short realisations of $B_4$ can be done in time
$n^9$ as follows: for every 6-tuple $F = (a, b, x_1, x_2, x_3, x_4)$
such that $G[F]$ has edge-set $\{x_1a, ax_2, x_2b, bx_3, bx_4\}$ and
for every 7-tuple $F = (a, b, x_1, x_2, x_3, x_4, x_5)$ such that
$G[F]$ has edge-set $\{x_1a, ax_2, x_2x_3, x_3b, bx_4, bx_5\}$, delete
$x_1, \dots, x_5$ and their neighbors except $a,b$.  In the resulting
graph, check whether $a$ and $b$ are in the same component.  The
answer is YES for at least one 7-or-6-tuple if and only if $G$
contains at least one short realisation of ${B_4}$.

Here is an algorithm for $\Pi_{B_4}$, assuming that the entry graph
$G$ has no short realisation of $B_4$.  For every 9-tuple $F = (a, b,
c, x_1, \dots, x_6)$ such that $G[F]$ has edge-set $\{x_1a, bx_2,
x_2x_3, x_3x_4, cx_5, x_5x_3, x_3x_6\}$ delete $x_1, \dots, x_6$ and
their neighbors except $a, b, c$.  In the resulting graph, run
three-in-a-tree for $a, b, c$.  It is easily checked that the answer
is YES for some 9-tuple if and only if $G$ contains a realisation of
$B_4$.

Let us prove that $\Pi_{B_5}$ is NP-complete by a reduction of
$\Gamma^3$ to $\Pi_{B_5}$.  Since by Theorem~\ref{th:delta3},
$\Gamma^3$ is NP-complete, this will complete the proof.  Let $(G, x,
y)$ be an instance of $\Gamma^3$.  Prepare a new graph $G'$: add four
vertices $x', x'', y', y''$ to $G$ and add four edges $xx', xx'', yy',
yy''$.  Since $\Delta(G) \leq 3$, it is easily seen that $G$ contains a
hole passing through $x, y$ if and only if $G'$ contains a realisation
of $B_5$.
\end{proof}

The proof of the theorem below is omitted since it is  similar
to the proof of Theorem~\ref{th:antena}.

\begin{theorem}
There is an $O(n^{14})$-time algorithm for $\Pi_{B_6}$ but $\Pi_{B_7}$
is NP-complete.
\end{theorem}

\subsection{Induced subdivisions of $K_5$}
\label{sec.k5}

Here, we study the problem of deciding whether a graph contains an
induced subdivision of $K_5$. More precisely, we put : $sK_5 = (\{a, b,
c, d, e\}, \emptyset, {\{a, b, c, d, e\} \choose 2})$.

\begin{figure}[h]
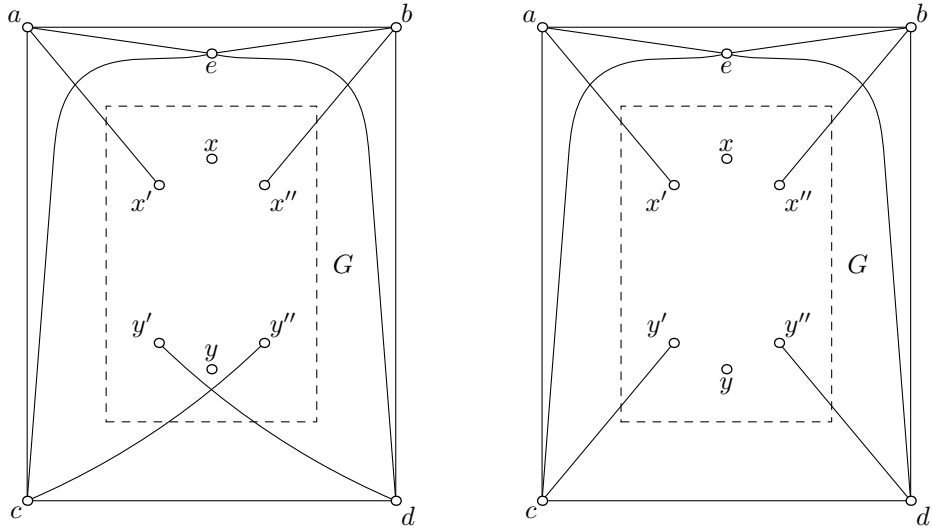

  \center
  \includegraphics{reductions.2}
  \rule{3em}{0em}
  \includegraphics{reductions.1}
  \caption{Graphs $G'$ and $G''$\label{fig.k5}}
\end{figure}

\begin{theorem}
  $\Pi_{sK_5}$ is NP-complete. 
\end{theorem}

\begin{proof}
We consider an instance $(G, x, y)$ of $\Gamma^3$.  Let us denote by
$x', x''$ the two neighbors of $x$ and by $y', y''$ the two neighbors
of $y$.  Let us build a graph $G'$ by adding five vertices $a, b, c,
d, e$.  We add the edges $ab, bd, dc, ca, ea, eb, ec, ed, ax', bx'',
cy'', dy'$.  We delete the edges $xx', xx'', yy', yy''$.  We define a
very similar graph $G''$, the only change being that we do not add
edges $cy'', dy'$ but edges $cy', dy''$ instead.  See
Figure~\ref{fig.k5}.

Now in $G'$ (and similarly $G''$) every vertex has degree at most 3,
except for $a, b, c, d, e$.  We claim that $G$ contains a hole going
through $x$ and $y$ if and only if at least one of $G', G''$ contains
an induced subdivision of $K_5$.  Indeed, if $G$ contains a hole
passing through $x, x', y', y, y'', x''$ in that order then $G'$
obviously contains an induced subdivision of $K_5$, and if the hole
passes in order through $x, x', y'', y, y', x''$ then $G''$ contains
such a subgraph.  Conversely, if $G'$ (or symmetrically $G''$)
contains an induced subdivision of $K_5$ then $a, b, c, d, e$ must be
the vertices of the underlying $K_5$, because they are the only
vertices with degree at least~$4$.  Hence there is a path from $x'$ to
$y'$ in $G\setminus \{x, y\}$ and a path from $x''$ to $y''$ in
$G\setminus \{x, y\}$, and consequently a hole going through $x, y$ in
$G$.
\end{proof}

\subsection{$\Pi_B$ for small B's}

Here, we survey the complexity $\Pi_B$ when $B$ has at most four
vertices.  By the remarks in the introduction, if $|V| \leq 3$ then
$\Pi_{(V, D, F)}$ is polynomial.  Up to symmetries, we are left with
twelve s-graphs on four vertices as shown below. 

For the following two s-graphs, there is a polynomial algorithm using
three-in-a-tree.  The two algorithms are essentially similar to those
for thetas and pyramids (see Figure~\ref{fig:ppt}).
See~\cite{chudnovsky.seymour:theta} for details.

\noindent
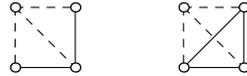
\begin{center} 
  \includegraphics{bigraphs.3} 
  \rule{3em}{0ex}
  \includegraphics{bigraphs.8} 
\end{center}

\noindent The next two s-graphs yield an NP-complete problem:

\noindent
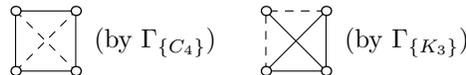
\begin{center} 
  \parbox[c]{1cm}{\includegraphics{bigraphs.4}} (by $\Gamma_{\{C_4\}}$)
  \rule{1em}{0ex}
  \parbox[c]{1cm}{\includegraphics{bigraphs.5}} (by $\Gamma_{\{K_3\}}$)
\end{center}

\noindent For the next seven graphs on four vertices, we could not get
an answer:

\noindent
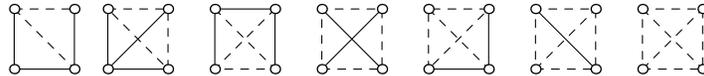
\begin{center}
\includegraphics{bigraphs.1}\rule{.5em}{0ex}
\includegraphics{bigraphs.6}\rule{1em}{0ex}
\includegraphics{bigraphs.7}\rule{1em}{0ex}
\includegraphics{bigraphs.9}\rule{1em}{0ex}
\includegraphics{bigraphs.10}\rule{1em}{0ex}
\includegraphics{bigraphs.11}\rule{1em}{0ex}
\includegraphics{bigraphs.12}
\end{center}

\noindent For the last graph represented below, it was proved recently
by Trotignon and Vu\v skovi\'c~\cite{nicolas.kristina:one} that the
problem can be solved in time $O(nm)$, using a method based on
decompositions.  

\noindent
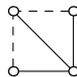
\begin{center}
\includegraphics{bigraphs.2}\rule{1em}{0ex}
\end{center}

In conclusion we would like to point out that, except for the problem solved
in~\cite{nicolas.kristina:one}, every detection problem associated with an
s-graph for which a polynomial time algorithm is known can be solved either
by using three-in-a-tree or by some easy brute-force enumeration.

\end{document}